\documentclass{article}

\usepackage[applemac]{inputenc}
\usepackage[T1]{fontenc} 
\usepackage[english]{babel}
\usepackage{amsmath,amssymb,amsthm}
\usepackage{natbib}
\usepackage{setspace}

\theoremstyle{definition}
\newtheorem{definition}{Definition}

\newtheorem{proposition}[definition]{Proposition}
\newtheorem{definition/proposition}[definition]{Definition/Proposition}

\title{Authorship and the Politics and Ethics of LLM Watermarks}

\author{Tim Räz\footnote{University of Bern, Institute of Philosophy,  L\"anggassstrasse 49a, 3012 Bern, Switzerland. E-mail: tim.raez@posteo.de}}

\begin{document}

\maketitle

\begin{abstract}
Recently, watermarking schemes for large language models (LLMs) have been proposed to distinguish text generated by machines and by humans. The present paper explores philosophical, political, and ethical ramifications of implementing and using watermarking schemes. A definition of authorship that includes both machines (LLMs) and humans is proposed to serve as a backdrop. It is argued that private watermarks may provide private companies with sweeping rights to determine authorship, which is incompatible with traditional standards of authorship determination. Then, possible ramifications of the so-called entropy dependence of watermarking mechanisms are explored. It is argued that entropy may vary for different, socially salient groups. This could lead to group dependent rates at which machine generated text is detected. Specifically, groups more interested in low entropy text may face the challenge that it is harder to detect machine generated text that is of interest to them.
\end{abstract}

\tableofcontents

\section{Introduction}
\label{sec:intro}

Large language models (LLMs) such as ChatGPT can generate high-quality texts, which creates opportunities, but also the possibility of abuse, including automated disinformation campaigns, spamming, plagiarism and academic cheating \citep{weidi2021}. Thus, it would be desirable to have a reliable method to distinguish text generated by machines and by humans  \citep{kirch2023,chris2023,aaron2023}. Such a method could also help to prevent information pollution, adverse consequences of machine generated text of any kind entering into the environment, such as machine generated text being used as training data for LLMs, which can degrade performance \citep{radfo2023}. Recently, so-called watermarking schemes \citep{kirch2023,chris2023} have been proposed to address these issues. Watermarking schemes inject a secret statistical signal into text generated by an LLM. This signal can be used to detect whether a text was generated by that LLM or not.

The goal of the present paper is to explore philosophical, political, and ethical ramifications of implementing and using watermarking schemes. To provide a conceptual backdrop, a definition of authorship that includes both machines (LLMs) and humans is proposed (Sec. \ref{sec:authorship}). The paper also provides an informal introduction to two watermarking schemes (Sec. \ref{sec:schemes}). On this basis, it is investigated how watermarking affects the usual determination of authorship (Sec. \ref{sec:author-certification}). It is argued that watermarking may provide private companies with sweeping rights of authorship determination, which is incompatible with traditional standards of authorship determination. A partial, institutional solution for this problem is proposed. The paper then discusses the possible ramifications of a common feature of watermarking mechanisms, viz. their entropy dependence (Sec. \ref{sec:entropy-groups}). In a nutshell, the reliability with which text can be detected as watermarked depends on there being an ``opportunity'' for injecting the watermark. This opportunity is measured by entropy. It is argued that entropy may vary for different, socially salient groups. If so, this could lead to different rates at which machine generated text is detected for the different groups. Groups more interested in low entropy text may face the challenge that it is harder to detect machine generated text that is of interest to them. The paper does not focus on the general question whether watermarks are beneficial, all things considered; this issue is discussed, e.g., by \citet{grinb2022}. Rather, the focus is on the (possibly adverse) consequences \emph{if} watermarking schemes currently available were put to use.

\section{Authorship}
\label{sec:authorship}

\subsection{The Definition of Authorship}

The following, preliminary definition of authorship provides a useful starting point.\footnote{The definitions discussed in this section are very similar to a ``narrow'' reading of authorship proposed by Foucault: ``I have discussed the author only in the limited sense of a person to whom the production of a text, a book, or a work can be legitimately attributed'' \citep[p. 216]{fouca1998}.} 

\begin{definition}
The author of a text is the person who wrote the text down.
\end{definition}

This definition is reasonably clear and operational, because it associates the author with the writing process.\footnote{It is also similar to definition I.1.a. of ``author'' in the Oxford English Dictionary, online version: ``The writer of a book or other work''.} However, it is easy to come up with counterexamples. First, the author cannot be a \emph{person}, because we want to discuss authorship of LLMs, and LLMs do not have personhood. The definition has to be weaker to encompass LLMs. Also, requiring that the author is \emph{one} person does not do justice to texts with multiple authors. Second, it is not necessary nor sufficient to require that an author \emph{wrote the text down}. It is not necessary because text can be produced in other ways: Dostojevsky dictated some of his novels, and we would not want to contest authorship because he did not write them down himself. It is not sufficient because the reproduction of text does not constitute authorship: You do not become the author of the bible by copying it. What seems necessary is that the text is created by the author in a suitable manner. However, the idea that the author creates the text is also problematic if it is not qualified. Any author is embedded in a cultural and linguistic context and is expected and even forced to draw on that context, cf. \citep{coeck2023}. What is more, the extent to which an author is expected to be the creator of a text depends on the kind of text. An author of a sonnet is not expected to created the form of the poem, but to work within the form.\footnote{\citet[p. 221]{fouca1998} questioned the idea of the author as a ``genial creator of a work'' and ``indefinite source of significations'', and pointed out the dependence of authorship (``authorship function'') on the kind of text (``discourse'') in question.} With these qualifications in mind, here is a refined version:

\begin{definition}
The author of a text is the entity who created relevant aspects of the text, where relevance is determined by the kind of text and the context.
\end{definition}

This definition accommodates some problems with the previous attempt. First, \emph{entity} does not require personhood, and it is compatible with the entity being a collective or set of people. Second, \emph{created} leaves open the specific process by which the text was generated. Third, \emph{relevant aspects of the text} accommodates the fact that not every aspect of a text is created by the author. It is also compatible with different people of a collective making different contributions, e.g., one author of a scientific paper providing guidance and ideas, and a second author formulating the text. This definition could still be considered problematic. \emph{Creation} may appear to refer to an exclusively human ability, requiring mental properties or a creative act. However, both mental properties and (intentional) actions are presumably absent in LLMs \citep{vanwo2024}. A weak sense of creation is needed, along the lines of the author producing relevant aspects of the text.\footnote{This is in line with a proposal by \citet{coeck2023}, according to which we do not need an authoritative origin, a creative author, for a text to be meaningful or significant.} But if we only require production instead of creation, a sense of novelty implicit in creation is lost. We thus have to require explicitly that the relevant aspects of the text are not merely reproduced. This yields the following definition.

\begin{definition}
The author of a text is the entity who produced relevant aspects of the text, where (i) production is not mere reproduction of the relevant aspects, and (ii) relevance is determined by the kind of text and the context.
\end{definition}

This definition of authorship is not general, it does not intend to capture all aspects of authorship, and is not supposed to provide necessary and sufficient conditions. First, it only captures the authorship of text, not of other modalities such as pictures. Second, it is supposed to encompass both humans and LLMs, which means that it is a \emph{weak} notion of authorship, avoiding a more narrow, anthropomorphic notion of intentional, creative authorship. It could be worried that the definition is too weak. For example, a random text generator qualifies, in principle, for authorship according to this definition. However, this is not problematic because the definition is only supposed to capture text that may have been produced by a human. For present purposes, authorship of LLMs is of interest because LLMs are able to produce text of relatively high quality. This excludes strings of symbols that have been produced by a random process, gibberish, and extremely incoherent text. The latter are simply not the cases we worry about. However, it is plausible that a certain degree of randomness can in fact play a role in producing high-quality text -- more on this below. Third, the definition is weak in that it does not presuppose agency, intentionality, or any kind of mental ability for authors. It is therefore also unsuitable as a definition of authorship that entails a legal or moral status for authors. \citet{porsd2023} note that many scientific journals have more stringent authorship criteria, including the ability of an author to approve of a text before publication, and to be accountable for their work, and these criteria exclude LLMs from authorship. 
\citet{vanwo2024} defend a stronger definition of authorship that entails, e.g., blameworthiness of authors for a text. These aspects are not part of the above definition. Note that this does not preclude the \emph{provider} of an LLM to be morally responsible for model output. Then, the relevance relation does a lot of work in this definition. There is no way around some degree of contextuality. In fact, contextuality is desirable because it opens the possibility of an influence of social norms on authorship, because social norms help determine what we expect of different kinds of text, including their ``originality''.

\subsection{The Determination of Authorship}
\label{sec:author-determin}

To determine authorship, we have to check which entity satisfies the conditions in the above definition of authorship. We can use internal or external evidence to determine whether an entity produced relevant aspects of a text without mere reproduction. Internal evidence is evidence contained in the text, while external evidence accompanies the text, e.g., metadata or physical properties of a hard copy.

Authorship can be proved with internal evidence if an entity has unique properties necessary to produce the relevant aspects of a text. For example, if a text contains information that only one entity has, then this entity is the author of the text; reproduction is excluded because no other entity has the information and thus the ability to produce the text. We will see below that watermarks for LLMs provide information of this kind: Watermarked text has a unique profile encoded by a key, and if this key is kept secret by the LLM provider, then the profile of the text could not have been produced by anyone else but the provider of the LLM, except with small probability.

Usually, determination of authorship is more circumstantial, in particular for human authors, because the internal evidence usually does not point to a unique entity. The default way to determine authorship is to accept the signatory of a text as the author. This default is usually not questioned, unless there is conflicting evidence. Conflicting evidence can take several forms. First, the properties of the signing author may not match the internal and external textual evidence. For example, assume I would publish one of Alice Munro's short stories situated in rural Canada under my name. If you knew me, you would justifiably question my authorship of that text, because I lack the knowledge necessary to write about this subject matter (I have never been to Canada), and in that style (I lack the necessary literary sophistication). In short, my properties do not match those of the author of the short story.
 
Second, two or more different entities may claim authorship with respect to the same aspects of a text, which constitutes an authorship conflict. In this situation, one method is to examine, as above, whether the properties of one of the entities match the internal evidence better than the other. Another method to resolve authorship conflicts is to examine external evidence. For example, in modern-day scientific priority disputes, we can look at the timestamps of papers on preprint servers. If the timestamp of one version is earlier than the other, and if there is no evidence that the text was created earlier by the other entity, then we ascribe authorship to the entity of the version with the earlier timestamp. Similarly, hard copies can be dated by examining their physical age and matched with authors in this way, and so on.

It could be asked whether it is possible that one and the same text was produced independently by two different entities. In some cases, this is possible. For example, there are examples where the same scientific discovery has been made by several scientists independently.\footnote{The discovery of supersymmetry in particle physics by three groups in the 1970s is an example \citep{rodri2010}.} However, independent production only happens on a coarse-grained level, that is, on the level of an idea expressed in a text, or with respect to a short piece of text, e.g., an equation. To see why, recall that it is very improbable that two entities produce a long random string of symbols independently of each other. The probability that two entities flip a fair coin 1000 times and come up with the same sequence is negligible. Of course, a text is not a random string. Still, the definition of authorship requires that the author produces the relevant aspects of a text independently of prior productions, and it is unlikely that two entities produce the same text independently of each other, given that there is sufficient opportunity to produce these relevant aspects. This informal idea is exploited by the watermarking schemes to be discussed below.

\subsection{The Authority of Authorship}

The fact that authorship and authority are related is apparent from the word stem of the word ``author''.\footnote{Definition II.5. of ``author'' in the Oxford English Dictionary, online version, reads: ``The person on whose authority a statement is made; an authority, an informant''.} In some cases, the authority of authorship has its origin outside of a text. If you already have power, then what you say has weight. A clear example of this is lawmaking. By signing the Executive Order on AI, President Biden put it into law. In this case, it is unlikely that the signatory of the law is also its author, except in a very general sense -- Biden may have ordered the EO to be written, indicating the general direction or spirit. In other cases, and for some kinds of text, the power of an author stems from their ability to produce that kind of text. We admire and pay novelists and scientists for their creative abilities, which is evident in their texts. In this case, authority arises at least in part from properties of the text itself.

Whenever there is an opportunity to gain power, some people will try to obtain it without putting in the work. An instructive example is ghostwriting, where the ghostwriter produces text for someone else without taking authorship credit, that is, they waive their right to be identified as the author of a text. The signatory, who takes credit, may pay the ghostwriter for their service, and also for their silence with respect to the actual authorship. Such a voluntary ``transfer'' of authorship can be tricky. If it is discovered that a signatory used a ghostwriter without disclosing it, they may lose some of the power they gained by claiming authorship. Depending on the kind of text, there is an expectation that authorship is faithfully disclosed. Of course, such norms can erode. For example, it has to be expected that celebrity memoirs are written with substantive help by editors or ghostwriters, even if this is not disclosed.

\section{Watermarking Schemes}
\label{sec:schemes}

\subsection{Preliminaries}

Watermarking schemes have two main parts, a watermark injection mechanism and a watermark detection mechanism. The goal of the injection mechanism is to inject a secret signal into text generated by an LLM, preferably such that the quality of the text does not suffer. The signal is secret in that it is only known to the LLM provider. The goal of the detection mechanism is to determine, given any text, whether or not the secret signal has been embedded. Given a text and the secret signal, the detection mechanism computes a score, which tells us how likely it is that the signal was embedded. If detection works well, the score allows a good separation between watermarked and non-watermarked text.

To be a bit more precise, an LLM computes a function $\mathcal{M}$ that takes a prompt (text) as an input and outputs a text in response.\footnote{The presentation in this section draws on \citet{kirch2023,chris2023}.} Syntactically, text is a sequence of tokens from a token set $\mathcal{X}$, the vocabulary, which may contain 50'000 tokens. The model starts by taking $\mathsf{prompt}$ as an input and computes the probability of the first token, $\mathcal{M}(\mathsf{prompt}) = p_0$, a distribution over the token set $\mathcal{X}$. The first token $x_0$ is then sampled from this distribution. Then the model repeats this process, using the prompt and the first token as an input. At time step $t$, the model computes $\mathcal{M}(\mathsf{prompt}, x_0, ..., x_{t-1}) = p_t$ and then samples $x_t$ from $p_t$. The process stops when a terminal token is drawn.

The injection mechanism modifies the sequence of token distributions $p_t$, injecting a pattern such that it can later be detected as watermarked, but such that the injection does not degrade the quality of the generated text too much or at all. The detection mechanism requires a text, possibly generated by an LLM, as input. The goal is to determine whether or not the text is watermarked. No prompt, probability distribution, or LLM access is necessary, because the prompt may be unknown, and because a text may not have been generated by an LLM. The detection mechanism can be private or public. If it is private, you do not know the signal, but you may query a detection API, which tells you whether or not a text is watermarked, without revealing the detection mechanism. If the detection mechanism is made public, this could be used to remove watermarks from text, and possibly to inject the watermark into text generated by a different LLM.

The degree to which text can be watermarked is limited. Informally, it is impossible to inject a watermark if certain aspects of a text, or even the entirety of a text, are predetermined. For example, if an LLM quotes a text, there is simply no opportunity to modify it to inject a watermark. Watermarking is only possible if there is an opportunity to do so. Formally, the degree to which we can inject a watermark depends the entropy of the distributions $p_t$ of text generation. (Shannon) entropy is a property of probability distributions and measures how much information we gain by sampling from that distribution; see the appendix for the definition. If an LLM reproduces a text verbatim, the probability distributions generating that text have their probability masses concentrated on the ``correct'' tokens that yield that text. For example, if the LLM quotes the beginning of Genesis I, it first outputs ``In the ...''. The distribution generating the next token has its probability mass concentrated on the word ``beginning''. In other words, the prediction is essentially deterministic. Such probability distributions have low or zero entropy. The more ``spread out'' the token probabilities, the higher the entropy. Entropy is maximal if a distribution is uniformly random. Thus, entropy measures the degree of randomness in text generation, which creates the opportunity to inject a watermark. Next we will see two schemes that put these informal ideas to work.\footnote{Note that while all works discussed below are technical, the scheme by \citet{kirch2023} is the most accessible. It is also recommended to look up the slides by \citet{aaron2023}, they provide a nice introduction to the main issues of watermarking schemes. The author first learned about watermarks from Aaronson's slides and his blog.}

\subsection{Kirchenbauer et al. 2023}

\citet{kirch2023} propose a watermarking scheme with the following properties: 1. Watermark injection can be applied to a trained model, no retraining is necessary. 2. Watermark detection does not presuppose access to the LLM or output probabilities. 3. Watermark detection works under certain assumptions if watermarked text is embedded in larger text. 4. The watermark can only be removed by changing a certain fraction of watermarked text. 5. Statistical tests for the confidence of watermark detection are formulated.

The authors first provide a simplified version of their scheme. The injection mechanism works as follows. Assume that the token $x_{t-1}$ has been predicted. For the prediction of token $x_t$, the vocabulary $\mathcal{X}$ is partitioned into equal-sized lists of tokens dubbed red and green. This partition is generated by a pseudorandom function $f$ that depends on $x_{t-1}$.\footnote{A pseudorandom function (PRF) is a (deterministic) mathematical function that ``looks like'' a random function. Technically, a PRF cannot be efficiently distinguished from a random function, cf. \citep{chris2023}. The PRF can be made to depend on a subset of previous tokens, which makes the scheme more robust.} The function $f$ always provides the same partition for the same input (token). Now the sampling from $p_t$ is modified: only green-list tokens are sampled, red-list tokens are excluded. The detection mechanism works as follows. Given a text we want to check for the watermark, we apply the function $f$ sequentially to each token, and determine, for every next token, whether that token is on the green or the red list. If a text is not watermarked, we expect green-list and red-list tokens to appear with the same probability. If a text, or a part of a text, is watermarked, we expect green-list tokens to appear more often than red-list tokens. For example, a text that is completely watermarked only contains green-list tokens. The probability that a non-watermarked text of length $T$ does not contain red-list tokens is $1/2^T$, which is very small even for texts of moderate length. This simplified scheme with a hard red list works, but it may degrade text quality in cases of low entropy. Kirchenbauer et al. illustrate the problem with the example of text containing the token ``Barack''. If, for the next token prediction, the token ``Obama'' is red-listed, which it is with probability $1/2$, then the text can never contain the string ``Barack Obama''. The token ``Obama'' has a high probability after ``Barack'', which means that the entropy of the prediction is low. Excluding the sequence ``Barack Obama'' from appearing in a text that contains ``Barack'' is undesirable.

To overcome this problem, Kirchenbauer et al. propose a soft version of the hard rule of red-list and green-list tokens. The degree to which tokens are red-listed depends on the so-called spike entropy of the generated text. Spike entropy has properties similar to Shannon entropy: It is low if the probability mass is concentrated on one or few outcome, and high if the probability has uniform distribution; see the appendix for the definition. The strength of watermark injection depends on the spike entropy of the distributions $p_t$.\footnote{The softness of the soft rule is controlled by a parameter delta, which is used to compute a modified version of the token probabilities. A second parameter, gamma, controls the fraction of tokens on the green list; above it was set to $0.5$.} If the spike entropy of $p_t$ is low, such that the probability mass is concentrated on one high-probability token (as in the case of ``Obama'' after ``Barack''), the soft rule does not have much of an effect on the probabilities. If spike entropy is higher and the probability mass is more evenly distributed, one fraction of tokens with similar probabilities get their probability boosted (green-listed), while the other fraction does not get a boost (red-listed). If spike entropy is sufficiently high for sufficiently many tokens, this creates a statistically discernible green-list bias, similar to the hard red list, which can be detected by applying the function $f$.

Kirchenbauer et al. then analyze the statistics of green-list tokens in watermarked and non-watermarked text. The idea is that watermarked text creates a distribution with a higher fraction of green-list items, given a minimal average spike entropy. Specifically, Kirchenbauer et al. derive a lower bound on the expected number of green list items and an upper bound on the variance of this number for watermarked text. The mean and variance of these numbers for text without watermark are easy to compute. Using standard approximations, they then analyze how much the statistics for watermarked and non-watermarked text overlap. Specifically, they compute how many false negatives (watermarked text that is not detected) we are likely to get for a given level of false positives (non-watermarked text detected as watermarked). Kirchenbauer et al. find that for a text of 200 tokens, average spike entropy of $0.8$, a green-list probability of $\gamma = 0.5$, and a false positive rate of the order of $10^{-5}$ (four standard deviations), the false negative rate is guaranteed to be below $1,4\%$.

An important question is whether the watermark injection affects the quality of generated text. Kirchenbauer et al. measure whether watermarked text deviates from the pre-watermarked version with a quantity called perplexity, which is closely related to entropy. They find that high-entropy and low-entropy output is not much changed by the watermark. However, there is a possible quality degradation for text with moderate entropy. 

The watermarking scheme can be used in public or in private mode. It is public if the key, the function that generates the partition, is public. If the key is public, everyone can reconstruct the green and red lists for each token. This can be used to remove the watermark. The key can also be kept private, and one can only query the key through an API. In this case, the number of queries has to be limited, because otherwise, the key can be reconstructed by tabulating answers to queries.

Kirchenbauer et al. also discuss possible attacks against this watermarking scheme, ways in which watermark can be removed and detection avoided. The strongest attack appears to be a so-called generative attack, in which an LLM is instructed to generate text that can be changed easily and such that the change removes some of the watermark. One example is a so-called emoji attack, in which the LLM is instructed to insert an emoji between every word. Such attacks are strong because they can possibly be automated easily -- just delete the emojis to obtain non-watermarked text. However, these attacks also presuppose a strong LLM that can follow the instruction, and they may increase the cost of text generation.

\subsection{Christ et al. 2023}
\label{sec:christ}

The watermarking scheme by \citet{chris2023} produces output that provably cannot be distinguished efficiently from non-watermarked output. This is the main difference to Kirchenbauer et al., who use a scheme that changes the output noticeably. The watermarking scheme by Christ et al. has three crucial properties. 1. \emph{Undetectability}: It is not possible to efficiently detect watermarked text without access to a secret key, even with query access to the model. 2. \emph{Completeness}: Watermarked text can be detected efficiently using the secret key, with negligible probability of not detecting watermarked text (false negatives), under an assumption about the entropy of model output. Detection also works on substrings. 3: \emph{Soundness}: The probability that non-watermarked text is misclassified as watermarked (false positives) is negligible.

As in Kirchenbauer et al., only text with a certain entropy can be watermarked. In fact, Christ et al. prove that the assumption about entropy is necessary: if one wants an undetectable watermarking scheme as defined by them, the model output can only be watermarked if it has a certain minimal entropy. Christ et al. extend their results to parts (substrings) of a text. At the core of this scheme is the cryptographic notion of efficient indistinguishability. Informally, two functions are efficiently indistinguishable if the probability that an efficient (polynomial time) algorithm can distinguish the two functions is very small.

Here is an informal description of key aspects of this scheme.\footnote{The focus here is on algorithms 3 (injection mechanism) and 4 (detection mechanism) in \citet{chris2023}.} First, Christ et al. show that it is sufficient to consider a binary alphabet. Then, they assume temporarily that one has access to true randomness. For the token $x_t$ at time $t$, start with the probability $p_t(1)$ that $x_t$ will be set to $1$ (recall that this is the binary case). Compare this probability to a uniformly random draw from $[0,1]$. If this draw is smaller than $p_t(1)$, let the output $x_t$ be 1, otherwise let $x_t$ be $0$. This method does not change the distribution $p_t$ for one response, because the draw is random. While this yields undetectability for a single $p_t$, if the method is applied to a sequence of token distributions, it yields a detectable signal. The detection mechanism uses a score that compares the tokens of a given text and the string of random draws of the same length. The score basically computes the Shannon information content (called empirical entropy by Christ et al.) of the random draw conditioned on the text; see the appendix. Now, if one computes the expectation of that score for watermarked text and non-watermarked text respectively, the expectation of non-watermarked text is $0$, while the expectation of watermarked text is essentially the Shannon entropy of the model output. In other words, if text is not watermarked, the random sequence is not embedded in the text; if it is watermarked, the random sequence can be detected if the entropy of text generation is sufficiently high.

This scheme needs some modifications to make it practical. As just described, it uses true randomness, and it can be only used once; if it is used repeatedly, its bias towards certain tokens starts to show. The issue with true randomness is solved by using a pseudorandom function (PRF); the use of the PRF is what yields the effective indistinguishability. To make the PRF repeatedly usable without revealing the statistical pattern it encodes, it uses the initial segment of tokens $(x_1, ..., x_{t-1})$ for the prediction of $x_t$. If enough entropy is used to generate this initial segment, it is essentially unique, because it is unlikely that other texts share this initial segment. The watermark injection mechanism is only activated once the entropy requirement for the initial segment is satisfied. To detect whether a text $(x_1, ..., x_n)$ is watermarked, the PRF is sequentially applied to initial segments $(x_1, ..., x_{t-1})$ for $t=2, ..., n$, and the pseudorandom signal $(u_t, ..., u_n)$ is computed for the remaining tokens $(x_t, ..., x_n)$. Then the detection score is computed from these two sequences to determine whether the signal $(u_t, ..., u_n)$ is embedded in the text $(x_t, ..., x_n)$. If this score is above a certain threshold for any of the $n-1$ initial segments, a watermark is detected, otherwise, not.

Christ et al. discuss two kinds of attacks on this scheme. Empirical attacks include post-processing of watermarked text, or prompts to generate text from which watermarks can be easily removed, including the emoji attack discussed by Kirchenbauer et al. . The second kind of attack, which could be called prefix specification attack, prompts the model one token at a time and thus prevents the formation of a watermark. Christ et al. prove that a sequence generated using prefix specification cannot be distinguished efficiently from non-watermarked model output. Thus, non-detection is guaranteed for this attack. This makes sense, because a watermark is only injected once sufficient entropy has accumulated in the initial segment. By restarting the prediction for every token, no entropy can accumulate, which prevents watermark injection. It is unclear to what extent such an attack can be prevented. Christ et al. point out that prefix specification is prohibited for ChatGPT, but is possible in the OpenAI Playground. They also note that the prefix specification attack may be prohibitive in terms of cost. The cost of generating a text of 8000 tokens with a single prompt costs less than \$1, while generating a text of the same length using the prefix specification attack costs about \$1000, because it requires a new prompt for every token.

\subsection{Other Relevant Works}

\citet{kudit2023} propose a watermarking scheme related to \citet{kirch2023}. This scheme preserves text quality, similar to \citet{chris2023}, and it is robust, which means, informally, that it is as hard to edit text to remove the watermark as it is to write the text without access to the LLM in question. \citet{kudit2023} do not explicitly rely on entropy for their statistical tests, but note that the so-called watermark potential, which measures how reliably one can distinguish watermarked from non-watermarked text, is aligned with entropy. \citet{aaron2023} proposes a watermarking scheme similar in spirit to \citet{chris2023}. This scheme also makes assumptions about the average entropy per token to provide guarantees, and produces output that is indistinguishable from non-watermarked output without access to a PRF. Aaronson developed this work at OpenAI, and details of the scheme are not public. Note that OpenAI has implemented this scheme, which means that it could be deployed, e.g., for ChatGPT, but is not deployed, as far as we know.

\citet{zhang2023} establish a theoretical limit of watermarking schemes. They show that so-called strong watermarking schemes are impossible. A strong watermarking scheme injects a watermark that cannot be removed efficiently without also degrading the quality of the watermarked text. Zhang et al. make two assumptions: First, it is possible to determine the quality of a model output for a given input with a quality oracle. Second, it is possible to perturb given outputs to obtain modified outputs with a perturbation oracle, where the perturbation preserves the quality of the output with a certain probability (along with other technical conditions). One crucial feature of this impossibility result is that removing the watermark means finding an alternative, high-quality response to a given prompt, which does not need to be semantically close to the original (watermarked) response. Zhang et al. show empirically that their result can be used to break the schemes by \citet{kirch2023,kudit2023}.

\section{Authorship Certification}
\label{sec:author-certification}

Watermarks add a new way to determine authorship, and it is important to consider how this meshes with existing ways to determine authorship. An authorship certification issued by an LLM provider states whether or not a certain text contains a watermark. In the positive case, such a statement can be interpreted as the claim that the text in question was authored by the LLM. Authorship certification can be based on watermark detection mechanisms, and detection can be handled in different ways. One important decision is whether the watermark detection mechanism, the key, is private or public. It will now be explored how these options affect authorship determination.

\subsection{Private Key}

A private key means that only the LLM provider has access to the watermark detection mechanism. The provider may allow for watermark detection via an API. The API may provide a yes/no answer, a confidence score, or highlight parts of a text containing the watermark.\footnote{Note that API access needs to be monitored, because otherwise, watermarking can be tampered with, either to steal the key, cf. \citep{kirch2023}, or by using a prefix specification attack.} Importantly, the above watermarking schemes rely on private keys for their statistical guarantees. A watermark injects a statistically unique property into a text, and its owner can use it to statistically verify whether that property was used to generate the text. If the key is private, the LLM provider thus obtains a private proof of authorship, because no one else could have known the unique property required to generate that text. If the key is kept private, then the LLM provider certifies authorship without providing public evidence for the certification. The authorship certification may be the result of a statistical test, but no proof or information that makes this test verifiable is provided \emph{to the public}. It is as if the LLM provider were saying: ``We will tell you whether or not our LLM wrote this, but we will not provide any evidence for our answer.'' Providing private companies with such sweeping rights of authorship certification seems unacceptable, because it provides these companies with the power to produce intentional false positives (certify LLM authorship of texts not produced by the LLM) and intentional false negatives (not certify LLM authorship of texts produce by the LLM). Below, both of these options are discussed.

The power of private authorship certification has limits. LLM providers cannot determine authorship freely, because the determination is a question of all available evidence, cf. Sec. \ref{sec:author-determin}. For example, if an LLM provider certifies authorship of a certain text, and someone else provides evidence that the text in question was written before the LLM began operating, this speaks against the authorship certification by the LLM provider. This also illustrates why private authorship certification is problematic if it is accepted as a brute fact: if there is no evidence that a text was \emph{not} generated by an LLM, a private key owner essentially has the authority to decide about the authorship of that text without providing any evidence.

Thus, if private authorship certification were accepted, it would have to be based on trust. Should LLM providers be trusted? The short answer is no. Generally speaking, the interests of private tech companies and the public are not aligned\footnote{Cases where a conflict of interest due to private authorship certification can arise are discussed by \citet{fairo2023}.}, and there is ample historical evidence of abuse of trust by tech companies, as witnessed, e.g., by the Cambridge Analytica scandal \citep{isaak2018,frenk2021}. If an LLM provider were to systematically abuse authorship certification, that is, provide intentional false positive and false negative authorship certifications, this abuse may eventually come to light, the LLM provider may lose public trust, and this may affect the provider's business. In the next two sections, it will be explored to what extent it could be in the business interest of LLM providers to provide false positive and false negative authorship certification if keys are private.

\subsection{Private Key, False Positive Certification}

It can be rational for an LLM provider to provide a false positive certification, that is, claim that the LLM is the author of a text not generated by the LLM. Assume that a high-level government official is hawkish about regulating LLMs, such that if the regulations suggested by the official were implemented, it would diminish the revenue of the provider significantly. Assume further that this official has written a dissertation in the past, independently of the LLM, but such that the text cannot be verified to have been authored independently of the LLM. Assume finally that a part of the dissertation is submitted to the watermark API of the provider. The provider can now certify that the submitted text was in fact authored by their model. This would undermine the official's credibility and could even lead to their resignation, undermining their ability to implement the regulation. Such a false positive certification would be rational because the expected loss due to regulation proposed by the official may be higher than the expected loss from being found out as producing a false positive certification. If the LLM provider has a good track record of authorship certification, this certification may have a high credibility.

This scenario is just one illustration of the possible damage of unchecked private key authorship certification. It is granted that the scenario is a bit contrived. However, there is precedent of government officials stumbling over plagiarism.\footnote{The former German defense minister Karl-Theodor zu Guttenberg is a prominent example \citep{suedd2011,ecker2015}.} Checks of plagiarism of academic works of government officials, journalists, and other high-profile individuals are now routine, which means that the above scenario, while unlikely, is a real possibility if API watermark detection were implemented.

\subsection{Private Key, False Negative Certification}

Now consider a scenario in which an LLM provider may use false negative certification.\footnote{\citet{kudit2023} call the fact that private keys make it necessary for detectors to trust LLM providers to actually inject the watermark the ``main inherent limitation of watermarking'' (p. 26).} One incentive for false negative certification is ghostwriting access to LLMs. Ghostwriting access means that the output of an LLM is not watermarked. Right now, all LLM access is ghostwriting access, because LLM output is not watermarked (as far as we know). However, in the future, if watermarking were either the de facto standard or even mandatory, ghostwriting access would be very valuable and could be sold at a premium. Thus, there is an incentive to sell ghostwriting access. It could be argued that ghostwriting access goes against the interest of LLM providers because it contributes to information pollution: LLM providers may want to watermark LLM output to be able to filter it out of their training data. However, LLM providers could use a second private key for text in ghostwriting mode, which is used internally, but not accessible through the watermark API.\footnote{Note that such a second, secret key of ghostwriting mode would not enable the company to detect generated text from other LLM providers. Thus, LLM providers would poison each other's well, unless they coordinate.} Note that prefix specification attacks are a kind of implicit ghostwriting mode, a guaranteed way of getting rid of watermarks. \citet{chris2023} write that the cost of these attacks may be prohibitive, cf. Sec. \ref{sec:christ}. However, depending on the scale of text generation, using a prefix specification attack is orders of magnitude less expensive than building a high quality LLM, which is a different way of circumventing watermarks. Also, the prevention of prefix specification attacks will presumably become harder as the number of LLM providers with high-quality output increases, because this increases the number of prompts that have to be rejected to prevent attacks.

\subsection{Public Key}

LLM providers could also publish their watermark detection mechanisms. In this case, whether or not a text was generated by an LLM is publicly verifiable, and it is not necessary to trust LLM providers to not provide false positive certifications. LLM providers could still provide ghostwriting access to select customers while watermarking all other text. Public keys create further problems. First, the watermarking schemes discussed above are based on private keys. Making keys public would destroy the statistical guarantees proved by \citet{chris2023}, and make it easier to break the scheme by \citet{kirch2023}. Thus, the watermarking schemes would lose some of their power. Second, third parties could try to use the detection mechanisms to reverse engineer the watermark injection mechanism. These third parties could then pollute the digital environment with text watermarked by the original LLM provider, thus hurting that provider.\footnote{This point is also discussed by \citet{fairo2023}.} Note that this presupposes that third parties have the capability to produce LLM output that can pass for output by the original LLM provider. In sum, public keys defeat, to some extent, the purpose of watermarks.

\subsection{Third Party Private Key}

A different option is to keep keys private, but to extend privacy to a third party, a watermarking clearing-house that is independent of LLM providers. This would shift some of the responsibility of authorship certification to a third party that does not have the incentives of LLM providers to provide false positive (and false negative) certifications. Third party certification could be done in different ways. One possibility is that all LLM providers devise their own watermarking mechanisms, but hand over their private keys to the third party for detection. A different possibility is that the third party creates and issues the private keys to the LLM providers. Watermark detection would be provided by the third party via API, possibly with information about which LLM provider generated a text.

This solution solves several problems. It would protect watermark functionality because keys would not be public. Then, LLM providers would no longer be in full control of authorship certification. The unique property that makes LLM output identifiable would be controlled by the third party, in particular if it is issued by the third party. As a consequence, intentional false positive certifications by LLM providers are no longer possible: Either the LLM provider injects the signal controlled by the third party into the text or not, but the signal would be detectable by the third party. A further advantage of this solution is that the third party could detect watermarks for many LLM providers. This is at least organizationally simpler than separate watermark detection for different providers. A third party private key does not prevent the possibility of ghostwriting mode, that is, LLM providers not watermarking text, even if it is mandatory. Note that there may be other, technical solutions interpolating between private and public key watermark detection. For example, \citet{fairo2023} propose a watermarking scheme with a public detection mechanism. This may lead to a technical, non-institutional solution to some of the problems discussed in this section.

\section{Entropy Differences Between Groups}
\label{sec:entropy-groups}

\subsection{The Idea}

Watermarking is entropy dependent, and entropy may vary for different groups. The watermarking schemes discussed above are entropy dependent in that the degree to which text can be watermarked, and thus also be detected as watermarked, depends on the entropy of text generation. Entropy dependence does not mean exactly the same in the two proposals, but they are similar.\footnote{In the scheme of Kirchenbauer et al., watermark injection depends on spike entropy. If the average spike entropy of text generation is higher, the false negative rate is lower for a given level of false positives. In the scheme of Christ et al., a watermark is only injected once the cumulative Shannon information content of generated text is sufficiently high. Christ et al. prove a completeness guarantee: if the cumulative Shannon information content of generated text is sufficiently high, the probability that watermarked text will not be detected as such is negligible. The \emph{expected} Shannon information content of text generation is Shannon entropy. Thus, on average, the completeness guarantee is easier to satisfy if the Shannon entropy of generation is higher.} Then, different kinds of text generated by LLMs may exhibit systematic entropy differences. Texts cover different topics, are written in different registers in one language, or in different natural languages. These categories, and evidence for entropy differences between some of them, will be explored below. Furthermore, interest in and use of different kinds of text may be correlated with group membership, including socially salient groups such as race, gender, class, age, etc. . The different kinds of text may function as proxies for group membership -- whether or not this is the case is an empirical question. Thus, different groups may encounter differences in how well generated text that interests them can be detected as watermarked, and this may result in benefits or disadvantages for these groups. Thus, watermarks may create a fairness risk of LLMs, if they are deployed \citep{weidi2021}. This argument will now be explored in more detail.

\subsection{Sources of Entropy Differences}

This section explores ways in which entropy differences between kinds of text may arise and be correlated with different groups. Before doing so, it should be stressed that entropy is formal property of text (generation) and should not be confused with a measure of text quality. If in doubt, recall that sequences of random numbers, which have maximum entropy, are not a compelling read. It is also stressed that which groups are associated with low or high entropy is an empirical question that has to be investigated in future work.

\paragraph{Topic Dependence}

Different groups may be interested in different topics, which leads to differences in text generation. Different topics, in turn, may be associated with different average entropies. If one group is more interested in generating text that is more constrained or standardized in form or content, and this group generates or uses more text of this kind, then text generation by and for this group may have lower entropy. To give an example, there is a correlation between education levels and migration.\footnote{One instance is an education gap of several years between first generation Mexican American men and non-migrant White American men in the US; cf. \citep[Ch. 10]{haas2023}. The same is true for some migrant groups in Europe. Note that these gap close quickly in the second to third generation.} Thus, in the context of education, first-generation migrants may, on average, be more interested in text generated for lower education levels than non-migrants. It is plausible that texts on lower education levels are more standardized, such that text generation for this level may have lower entropy than for higher education levels. Again, this last part is a conjecture in need of empirical confirmation. If the conjecture is true, then, in the context of education, the texts that are of interest to migrants, for generation and for reading, have lower average entropy levels, and thus, potentially, a higher false negative rate.

\paragraph{Register Dependence}

The language register, or style, is a potential source of entropy differences between groups. People use language differently, because they choose a certain tone, high or low, literary or informative, or because they have certain means of expression at their disposal, which is shaped by who they are, including membership in different groups. Do registers of natural language yield different entropies? At a first glance, this question does not makes sense, because entropy is a property of probability distributions, and we do not have such distributions for text in natural language that was not generated by an LLM. However, we can estimate such distributions for any text, and then calculate the entropies of the distribution estimates.\footnote{There is a tradition, going back to Shannon, of trying to estimate the entropy of texts and different languages. A simple method is to let humans guess the next letter of a given text. A method that appears to yield good estimates is to let humans place bets / gamble on the next letter of a given text. The optimal gamble is proportional to the distribution of the next letter, cf. \citep[Ch. 6]{cover2006}.} For example, \citet{konto1997} estimates the entropy of three text corpora: The King James Bible, two novels by Jane Austen, and two novels by James Joyce. Kontoyiannis found the following estimates of entropy (bits-per-character) for these authors: 0.92 for the Bible, 1.78 for Jane Austen, and 2.12 for James Joyce. These results seem to agree with intuition. The Bible is a formulaic and at times repetitive text with a relatively small vocabulary. Jane Austen has a more literary style, mirroring literary and social conventions of the 19th century. James Joyce, finally, has a modern literary style, with high variance down to the level of vocabulary.\footnote{Again, the reader is cautioned against interpreting these entropy estimates as somehow capturing the quality of these texts.}

This suggests that different registers or styles have different entropies in natural language. It could be thought that this does not matter for generating text with LLMs, because everybody has access to the same models with the same capabilities, and thus the same entropy rates. One reason why register matters is that different registers are associated with different topics, such that different topics may yield different entropy rates, mediated by the entropy of the style of a topic. A second reason is personalization. Assume that someone wants to use an LLM to produce text that fits their personal style, e.g., by providing the LLM with samples of text written by them. A skilled LLM will mirror this style, which means that the output will have the entropy of that style. It is well-known that LLMs have the capability of adopting different styles.\footnote{Although these abilities presumably do not yet match human-level capabilities, cf. \citep{pu2023}.} Thus, differences in entropy between different styles, different sociolects etc., may impact entropy of LLM output, and thus watermark quality. To give an example, I am a non-native speaker of English. If I want to use an LLM to produce personalized text, and if there is an entropy difference between English produced by native and non-native speakers, this will lead to an entropy difference between the output produced for me, and the output for the same prompt in the style of a native speaker. Thus, either I use an LLM to produce output that does not sound like me, or there may be an entropy difference between text produced in my style and the style of a native speaker. And this may yield a difference in watermark quality.\footnote{The difference between the entropy of native and non-native English speakers appears to be not very well studied. There is some evidence for differences depending on the first language of English learners, measured as Kolmogorov complexity \citep[Sec. 6.2.3]{ehret2016}, but the results are not very reliable due to data availability, and should thus be taken with caution, as the author of the study points out.}

\paragraph{Natural Language Dependence}

Entropy differences may arise because different languages have different entropies. There is empirical evidence that different languages do in fact have different entropies. \citet{kople2023} provide evidence that languages with more speakers (so-called large languages) have higher entropy.\footnote{\citet{kople2023} examined what they call the prediction complexity (the average Shannon entropy of next-token prediction) of different natural languages by training language models on large text corpora. Note that the issue of natural language complexity appears to be somewhat contested in computational linguistics.} Thus, languages like Mandarin and English may have higher entropies than small languages or dialects. Note that other factors may influence the use of LLMs in different languages. In particular, the availability of high-quality LLMs in a given language depends on the existence of large text corpora, which is limited for small languages.

\paragraph{Translation}

If there are entropy differences between natural languages, one may wonder about the entropy of translation, and the use of translation by different groups. Machine translation is a different task than text generation, because it presupposes a given text to be translated. However, if, say, high quality machine translation from language 1 to language 2 is combined with a high quality generative model in language 1, the combination can be viewed as a high quality generative model in language 2, and we can ask about the entropy of translation in comparison to a generative model in language 2.

Here are some hypotheses to be investigated. First, comparing the entropy of text generation with translation, the entropy of translation may be lower, other things being equal, because in translation, the output is constrained to be the translation of text that already exists, while generating novel text is not equally constrained.\footnote{\citet{arora2023} call tasks like machine translation ``strongly conditioned generation tasks'', and tasks like dialogue generation ``open-ended generation tasks''. Note that we can expect entropy to be lower for a more literal translation, and higher for a freer version or a ``recreation'' of a work in a different language.} It is possible that the text to be translated already contains a watermark, but it is not clear whether watermarks are robust under translation. Thus, people relying on translations in any language may be confronted more with non-watermarked text. Second, if different natural languages have different base entropies, the entropies of translation into these languages may also differ. Also, there may be more translation into smaller languages in general, because more text and information may be available in larger languages. Thus, people using text in small languages may be confronted more with non-watermarked text, because small languages have lower entropy.

\subsection{Consequences of Entropy Differences}

Now consider possible consequences of entropy differences between groups. People from different groups may both benefit and be harmed by either the presence or absence of watermarks, depending on their goals and roles. Distinguish, first, between contributors and recipients of text. Contributors are people who write and generate text, while recipients are people who read, consume, use text. For contributors, the goal may be to generate text for fair use or to generate text with the intent to deceive about the machine origin of the text. If the intent is to deceive, the goal may be to plagiarize, that is, to pass off generated text as written by a human, or to evade detection of generated text for misinformation or spamming. Contributors can also be human writers who may want to be recognized for their work. Finally, recipients can be readers of texts written for a particular community, or engineers collecting text to train an LLM.
 
To explore the consequences of entropy differences, consider a simple scenario. Assume that there are two text pools, $T_{low}$ and $T_{high}$, sets of text in which two different groups are interested, but which serve the same or similar functional roles for the two groups -- think of, e.g., different education levels. There may be an overlap between $T_{low}$ and $T_{high}$, but it is also assumed that they are relevantly different, in that text generation for $T_{low}$ has, on average, lower entropy than generation of text for $T_{high}$. Assume further that both groups have the same level of false positives, because false positives are very undesirable and should be kept as low as possible for both kinds of text. In this scenario, watermarked text in $T_{low}$ will be undetected at a higher rate than generated text in $T_{high}$, that is, the false negative rate in $T_{low}$ is higher than the f.n.r. in $T_{high}$. How are people in the two groups affected by this difference?

\begin{itemize}

\item \textit{Generate text for fair use:} If you generate text for $T_{low}$, you contribute to informational pollution at a higher rate than if you contribute to $T_{high}$, because the text you generate is harder to filter out, even if you have no malicious intent. Thus, if you care about the integrity of $T_{low}$, you may be more hesitant to generate text for $T_{low}$ than for $T_{high}$, other things being equal.

\item \textit{Generate text for plagiarism:} If you generate text for $T_{low}$ with the intent to plagiarize, passing yourself or others off as authors, your chances of getting caught are lower than if you generate text for $T_{high}$. Whether or not this makes plagiarizing for $T_{low}$ worthwhile depends on the cost of getting caught, and on the gains of getting away with plagiarism in $T_{low}$ and $T_{high}$. The cost of getting caught may be prohibitive, such that plagiarizing is not worth it even given the better odds for $T_{low}$. The gain from plagiarizing for $T_{high}$ may also be higher, e.g., if the prestige gained from contributing to $T_{high}$ is higher.

\item \textit{Generate text for misinformation / spamming:} If you generate text for spamming or misinformation, the expected cost of getting caught is lower than for plagiarism, because you will usually not be identifiable as having generated text. However, you do care about evading spam filters or content moderation, so generating text for $T_{low}$ is more attractive, other things being equal, than generating text for $T_{high}$. However, there may be a difference in gains to be made from malicious contributions to $T_{low}$ and $T_{high}$, which may change incentives. 

\item \textit{Write text:} If you contribute text to $T_{low}$ as an author, you contribute to an environment that may already have stronger informational pollution than $T_{high}$, meaning that your contribution may be drowned out by generated text. Also, your contribution may run the danger of being taken to be LLM output, because people are aware of higher false negative rates, such that the trust in your contribution is lower than in $T_{high}$. So, there may be an incentive to contribute to $T_{high}$, because trust is higher and taking authorship credit is easier. Again, whether or not this is the case depends on comparative gains.

\item \textit{Read text:} If you are a reader of $T_{low}$, you read in an environment where information about authorship is less reliable -- you have to expect that text you read is machine generated but not flagged as such. This can lead to a loss of trust in texts in $T_{low}$. Also, because of the incentives of building models based on data from  $T_{low}$, the quality of machine generated text in $T_{low}$ may be lower (see the next point).

\item \textit{Use text for LLM training:} Using generated text to train LLMs may degrade their quality. If you are an engineer and you know that there is a higher proportion of generated text in $T_{low}$ that you cannot get rid of using watermark detection, you may be more hesitant to use text in $T_{low}$ for training. As a consequence, LLMs may be less skilled to produce text for $T_{low}$ than for $T_{high}$, because there is an incentive not to use it for training.

\end{itemize}

In sum, some people, in particular those with malicious intent, such as spamming and misinformation, may profit from a low entropy text generation environment -- they can get away with more. For other people, other things equal, being in a low entropy text generation environment may cause harm, because information about authorship is less reliable. Note, incidentally, that the above analysis can also be read as an argument in favor of watermarks, because in an environment without watermarks, the false negative rate is higher than in an environment in which at least some generated text can be detected with watermarks.

\section{Conclusion}

In this paper, a weak, contextual notion of authorship that can accommodate both humans and LLMs was proposed. It was argued that authorship is determined on the basis of evidence about who produced relevant aspects of a text independently of prior productions. Based on an examination of existing watermarking schemes, it was argued that leaving authorship certification in the hands of private companies should be avoided. Instead, authorship certification should be handed over to a third party. Then the paper examined the possible consequences of the entropy dependence of watermarks for different groups, and it was argued that, depending on roles, there may be differential benefits from watermarks.

The arguments in this paper are incomplete and in need of further investigation. First, the examination was based on watermarking schemes that are now available. While it is plausible that future watermarking schemes share some of their properties, it is possible that future schemes use different kinds of entropy for injection. It is also possible that future schemes will provide publicly verifiable detection mechanisms while protecting watermarking functionality. Second, the above considerations of entropy differences between groups are by and large conjectures that need empirical verification. Finally, it is not claimed that watermarks are problematic and should be avoided, generally speaking. Watermarks, like any other technology, may be beneficial, all things considered, while at the same time providing those in power with disproportionate benefits, and possibly having a disparate impact on different groups. How to mitigate these issues will have to be addressed in future work.

\appendix

\section{Definitions of Entropy}
\label{sec:entropy}

The definitions of Shannon entropy and information content are standard, cf. \citep{macka2003,cover2006}.

\begin{definition}
The \emph{Shannon information content} $h(x)$ of an outcome $x$ of a discrete random variable $X$ is defined by:
\begin{equation}
h(x) = log \frac{1}{p(x)}.
\end{equation}
It is a measure of how much information we gain when we observe the outcome $x$. \citet{macka2003} uses this name; \citet{chris2023} call it empirical entropy.
\end{definition}

\begin{definition}
The \emph{Shannon entropy} $H(X)$ of a discrete random variable $X$ is defined by:
\begin{equation}
H(X) = - \sum_{x \in \mathcal{X}} p(x)\,log\, p(x).
\end{equation}
It can also be written as $H(p)$. Shannon entropy is the average Shannon information content of the outcomes of $X$, as noted in \citep{macka2003}.
\end{definition}

\begin{proposition}
$H(X) \geq 0$, and $H(X) = 0$ if the probability mass is concentrated on one outcome, i.e., $p(x) =1, p(x') = 0$ for $x' \neq x$.
\end{proposition}

\begin{proof}
$0 \leq p(x) \leq 1$ implies $log \frac{1}{p(x)} \geq 0$. For the second part, note that $0\, log\, 0 = 0$, which follows from $lim_{x \rightarrow 0}\, x\, log\, x = 0$; cf. \citep{cover2006}.
\end{proof}

\begin{proposition}
$H(X) \leq log |\mathcal{X}|$, where $|\mathcal{X}|$ is the size of the range of $X$, with equality if and only if $X$ has a uniform distribution over $\mathcal{X}$.
\end{proposition}

\begin{proof}
\citep[Theorem 2.6.4]{cover2006}.
\end{proof}

\begin{definition}
The \emph{average Shannon entropy} of a sequence $X_1, X_2, ..., X_n$ of random variables is:
\begin{equation}
\frac{1}{n} \sum_i H(X_i)
\end{equation}
\end{definition}

\begin{definition}
\citep[Def. 4.1.]{kirch2023} The \emph{spike entropy} of a distribution $p(x)$ with modulus $z$ is 
\begin{equation}
S(p, z) = \sum_{x \in \mathcal{X}} \frac{p(x)}{1 + zp(x)}.
\end{equation}
Kirchenbauer et al. note that spike entropy is minimal if the probability mass is on one outcome, and maximal for uniform distributions, like Shannon entropy. The \emph{average spike entropy} of a sequence of distributions $p_1, ..., p_n$ is:
\begin{equation}
\frac{1}{n}\sum_i S(p_i, z). 
\end{equation}
\end{definition}

\begin{definition}
\citep{chris2023} The detection score for one bit $x_t$ and one random draw $u_t$ from $[0,1]$ is:
\begin{equation}
s(x_t, u_t) = 
	\begin{cases}
	\text{ln}\, \frac{1}{u_t} & \text{if } x_t = 1, \\
	\text{ln}\, \frac{1}{1 - u_t} & \text{if } x_t = 0.
	\end{cases}
\end{equation}
The detection score for a text (string) $x = (x_1, ..., x_n)$ and a sequence of random draws $(u_1, ..., u_n)$ is:
\begin{equation}
c(x) = \sum_{i=1}^n s(x_i,u_i)
\end{equation}
\end{definition}

\end{document}